%% file: main.tex
\documentclass[11pt,letterpaper]{article}

\usepackage[margin=1in]{geometry}
\usepackage{latexsym,graphicx,amssymb}
\usepackage{amsmath,enumerate}
\usepackage{bbm}
\usepackage{float}
\usepackage{epsfig}
\usepackage{xspace}
\usepackage{paralist}
\usepackage{enumerate}
\usepackage{cases}
\usepackage{caption}
\usepackage{multicol}
\usepackage{graphicx}
\usepackage{xcolor}
\usepackage{tikz}
\usepackage{ifthen}
\usepackage{algorithm}
\usepackage[noend]{algpseudocode}
\usepackage{amsthm}
\usepackage{tcolorbox}
\usepackage{dsfont}

\usepackage{caption}
\usepackage{subcaption}

\usepackage{url}

\usepackage[linesnumbered,ruled,vlined,algo2e]{algorithm2e}

\usepackage{siunitx}
\usepackage[section]{placeins}

\SetKw{Return}{return}%

\newtheorem{definition}{Definition}[section]

\newtheorem{theorem}{Theorem}[section]
\newtheorem*{theorem*}{Theorem}

\newtheorem{observation}{Observation}[section]

\newtheorem{lemma}{Lemma}[section]
\newtheorem*{lemma*}{Lemma}

\newtheorem*{claim*}{Claim}

\newtheorem{remark}{Remark}[section]

\input{notations}

\title{Revisiting Ranking for Online Bipartite Matching with Random Arrivals: the Primal-Dual Analysis}
 \author{
 Bo Peng \thanks{ITCS, Shanghai University of Finance and Economics, \texttt{ahqspb@163.sufe.edu.cn}}
 \and
 Zhihao Gavin Tang \thanks{ITCS, Shanghai University of Finance and Economics, \texttt{tang.zhihao@mail.shufe.edu.cn}}
 }

\date{}

\begin{document}

\maketitle

\begin{abstract}
We revisit the celebrated Ranking algorithm by Karp, Vazirani, and Vazirani (STOC 1990) for online bipartite matching under the random arrival model, that is shown to be $0.696$-competitive for unweighted graphs by Mahdian and Yan (STOC 2011) and $0.662$-competitive for vertex-weighted graphs by Jin and Williamson (WINE 2021).

In this work, we explore the limitation of the primal-dual analysis of Ranking and aim to bridge the gap between unweighted and vertex-weighted graphs. 
We show that the competitive ratio of Ranking is between $0.686$ and $0.703$, under our current knowledge of Ranking and the framework of primal-dual analysis.
This confirms a conjecture by Huang, Tang, Wu, and Zhang (TALG 2019), stating that the primal-dual analysis could lead to a competitive ratio that is very close to $0.696$.
Our analysis involves proper discretizations of a variational problem and uses LP solver to pin down the numerical number.
As a bonus of our discretization approach, our competitive analysis of Ranking applies to a more relaxed random arrival model. E.g., we show that even when each online vertex arrives independently at an early or late stage, the Ranking algorithm is at least $0.665$-competitive, beating the $1-1/e \approx 0.632$ competitive ratio under the adversarial arrival model.
\end{abstract}

\section{Introduction}
\input{intro}

\section{Preliminaries}
\label{sec:prelim}
\input{prelim}

\section{Discretization}
\label{sec:discretization}
\input{discretization}

\input{upper}

\section{Numerical Results}
\label{sec:num}
\input{numerical}

\section{Discussion}
\input{discussion}

\bibliographystyle{plain}
\bibliography{matching.bib}

\end{document}

%% file: notations.tex
\newcommand{\be}{\begin{equation}}
\newcommand{\ee}{\end{equation}}
\newcommand{\beq}{\begin{equation*}}
\newcommand{\eeq}{\end{equation*}}

\newcommand{\E}{\mathbb{E}}

\newcommand{\Xcomment}[1]{{}}

\newcommand{\eqdef}{\overset{\mathrm{def}}{=\mathrel{\mkern-3mu}=}}
\newcommand{\vect}[1]{\ensuremath{\mathbf{#1}}}

\newcommand{\RN}[1]{\textup{\uppercase\expandafter{\romannumeral#1}}}

\newcommand\restr[2]{{\left.\kern-\nulldelimiterspace #1 \vphantom{\big|} \right|_{#2} }}

\newcommand{\dd}{\mathrm{d}}

%% file: intro.tex
In a seminal work, Karp, Vazirani, and Vazirani~\cite{stoc/KarpVV90} introduced the online bipartite matching problem and the Ranking algorithm. 
Consider a bipartite graph whose two sides of the vertices being offline and online respectively. The offline vertices are known in advance and the online vertices arrive one by one in a sequence. Upon the arrival of an online vertex, its incident edges are revealed and the algorithm chooses immediately one of its unmatched neighbor (if any) to match to. The goal is to maximize the size of the selected matching.
The Ranking algorithm works as the following: it uniformly at random permutes all offline vertices in advance and then matches each online vertex to the first unmatched neighbor with respect to the permutation.
The algorithm is shown to be $\left(1-\frac{1}{e}\right)$-competitive and this ratio is the best possible in the worst case.

Since then, online bipartite matching has become an exceptionally active and productive topic within the online algorithms community.
A rich variants of the online bipartite matching problems have been proposed and Ranking is mostly the first algorithm to be investigated for its viability.
Aggarwal et al.~\cite{soda/AggarwalGKM11} generalized the Ranking algorithm to vertex-weighted graphs and proved an optimal $\left(1-\frac{1}{e}\right)$ competitive ratio. Karande, Mehta, and Tripathi~\cite{stoc/KarandeMT11}, and Mahdian and Yan~\cite{stoc/MahdianY11} examined the performance of Ranking in the of random arrival model and established a competitive ratio of $0.696$. Huang et al.~\cite{talg/HuangTWZ19} studied the combination of the two settings, i.e., for vertex-weighted graphs with random arrivals and proved that the Ranking algorithm achieves a competitive ratio of $0.653$. The competitive ratio is later improved by Jin and Williamson~\cite{wine/JinW21} to $0.662$.

The analysis of Ranking has also been simplified by a series of work~\cite{sigact/BenjaminC08,soda/GoelM08,soda/DevanurJK13, sosa/EdenFFS21}. Remarkably, the randomized primal-dual analysis by Devanur, Jain, and Kleinberg~\cite{soda/DevanurJK13} unified the competitive analysis of Ranking on unweighted and vertex-weighted graphs. 
The framework has been further developed and applied to other online matching problems and an important merit of the primal-dual analysis is its intrinsic robustness for vertex-weighted graphs. Many variants of the online matching problem that employ the primal-dual analysis thus share the same state-of-the-art competitive ratio for unweighted graphs and vertex-weighted graphs.
One exception is the significant gap between the state-of-the-art competitive ratios for unweighted and vertex-weighted graphs in the random arrival model. 
The analysis of Mahdian and Yan is not primal-dual based and it remains unclear whether their $0.696$ competitive ratio carries over to the vertex-weighted case, since Jin and Williamson observed a $0.668$ barrier of the previous randomized primal-dual analysis of Huang et al.~\cite{talg/HuangTWZ19}.

In addition, the primal-dual analysis admits a folklore economic interpretation, refer to e.g. \cite{sosa/EdenFFS21,corr/Hartline22}. In fact, it simulates a dynamic pricing process. Given a bipartite graph $G = (L, R, E)$, consider the offline vertices $L$ as goods and the online vertices $R$ as utility-maximizing unit-demand buyers. 
Then in the primal-dual analysis, each dual variable $t_v$ for $v\in L$ corresponds to the price of good $v$; and each dual variable $t_u$ for $u \in R$ corresponds to the utility of buyer $u$.

In this work, we revisit the randomized primal-dual analysis of Ranking for online bipartite matching with random arrivals and aim to bridge the existing gap between the unweighted case and the vertex-weighted case. We formalize the best factor-revealing optimization using our current knowledge of Ranking\footnote{We do not define accurately what we mean by the best factor-revealing optimization. We remark that our formulation utilizes the same structural properties of Ranking as Mahdian and Yan and shall make further remarks regarding this statement in Section~\ref{subsec:upper}.} under the randomized primal-dual framework, and explore the limitation of this approach. We provide a lower bound of $0.6862$ and an upper bound of $0.7027$. 
\begin{theorem}
The Ranking algorithm is at least $0.6862$-competitive for the vertex-weighted online bipartite matching problem with random arrivals.	
\end{theorem}
\begin{theorem}[Informal, refer to Theorem~\ref{thm:opt}]
Under the current randomized primal-dual framework, 	it is impossible to establish a competitive ratio better than $0.7027$ of Ranking in the random arrival model, even for unweighted graphs.
\end{theorem}

Our positive result confirms a conjecture of Huang et al.~\cite{talg/HuangTWZ19}, stating that the primal-dual analysis could lead to a competitive ratio that is very close to $0.696$.
Our analysis involves proper discretizations of the optimization problem and uses LP solver to pin down the numerical number. 
As a bonus of our discretization approach, our positive result holds with respect to a family of \emph{independent random arrival} model, that smoothly bridge the worst-case arrival order and the classical random arrival model.

\paragraph{Independent Random Arrival.}
Let there be $n$ online vertices and $m$ stages from early to late. We assume that each online vertex arrives independently at one of the stages. Among the vertices that arrive at the same stage, their relative order is arbitrary. Notice that when $m=1$, the model degenerates to the worst-case arrival order; when $m \to \infty$, the model becomes the classical random arrival model as the probability of having two vertices at the same stage is $o(1)$.

\begin{theorem}
When each online vertex arrives independently at early or late stage (i.e., $m=2$), the Ranking algorithm is at least $0.6656$-competitive. 
When each online vertex arrives independently at an early, middle, or late stage (i.e., $m=3$), the Ranking algorithm is at least $0.6763$-competitive.	
\end{theorem}

The minimal independent random arrival that goes beyond the worst-case arrival is when $m=2$. Even for this case, our $0.6656$ competitive ratio improves the previous best $0.6629$ ratio of Jin and Williamson~\cite{wine/JinW21}. And when $m=3$, our result surpasses the $0.6688$ barrier of previous approaches. Refer to Section~\ref{subsec:lower_ratio} Table~\ref{table:lower} for more competitive ratios of different $m$.

Finally, we would like to remark that the current best ratios (i.e., $0.6862$ and $0.7027$) are limited by computational power, but not our approach.
We also observe that the current randomized primal-dual analysis leads to numerically close competitive ratios as those ratios of the strongly factor-revealing linear program of Mahdian and Yan\footnote{The current state-of-the-art $0.696$ competitive ratio is also limited by computational power. Mahdian and Yan showed that their approach would lead to a ratio between $0.6961$ and $0.7014$.}. More details are provided in the discussion section.

\subsection{Further Related Works}
Besides the series of works simplifying the competitive analysis of Ranking, Uriel Feige~\cite{corr/Feige18} improved the analysis of Ranking for lower order terms. Moreover, the Ranking algorithm has also been studied for other variants of maximum matching problems. In the fully online matching setting, Huang et al. established a $0.567$-competitive ratio for bipartite graphs~\cite{soda/HPTTWZ19} and a $0.521$-competitive ratio for general graphs~\cite{jacm/HuangKTWZZ20}. In the oblivious matching problem, Chan et al.~\cite{sicomp/ChanCWZ18,talg/ChanCW18} established a $0.526$-competitive ratio for general graphs.

The randomized primal-dual approach has been further developed beyond the classical online bipartite matching problem, including edge-weighted online bipartite matching~\cite{jacm/FahrbachHTZ22, focs/BlancC21}, AdWords~\cite{jacm/MehtaSVV07,focs/HuangZZ20}, fully online matching~\cite{jacm/HuangKTWZZ20,focs/HuangTWZ20,soda/HPTTWZ19}, oblivious matching~\cite{jacm/TangWZ23}, online matching with stochastic rewards~\cite{stoc/HuangZ20,wine/HuangJSSWZ23}, and streaming submodular matching~\cite{soda/LevinW21}, etc. Remarkably, the state-of-the-art competitive ratios of these settings are all primal-dual based analysis.

Beyond the random arrival order setting, there is a line of work studying online stochastic matching. Feldman et al.~\cite{focs/FeldmanMMM09} introduced the known i.i.d. arrival model in 2009, earlier than the time when the random arrival model is studied, and provided the first competitive algorithm that surpasses the $1-1/e$ barrier. The competitive ratio is later improved by a series of works~\cite{esa/BahmaniK10,algorithmica/BrubachSSX20a,mor/JailletL14,mor/ManshadiGS12,stoc/HuangS21,stoc/TangWW22,stoc/HuangSY22}. The stochastic setting is further extended to vertex-weighted graphs~\cite{mor/JailletL14,stoc/HuangS21,stoc/TangWW22,stoc/HuangSY22}, edge-weighted graphs~\cite{algorithmica/BrubachSSX20a,wine/HaeuplerMZ11,wine/FengQWZ23,yan2024edge}, and non-identical arrival setting~\cite{stoc/TangWW22}. Recently, there is a growing interest in designing computationally efficient algorithms against the optimal online algorithm~\cite{ec/PapadimitriouPS21,icalp/SaberiW21,ec/BravermanDL22,corr/NaorSW23} for stochastic matching.

%% file: prelim.tex
Let there be an underlying vertex-weighted bipartite graph $G=(L,R,E)$, where $L$ denotes the set of offline vertices and $R$ denotes the set of online vertices. Each offline vertex $v \in L$ is associated with a non-negative weight $w_v$. 

\paragraph{(Vertex-Weighted) Ranking.} The vertex-weighted Ranking algorithm is characterized by a function $f:[0,1]^2 \to [0,1]$. Each offline vertex $v$ draws in advance an independent random variable $y_v \in [0,1]$ uniformly at random. We refer to $y_v$ the rank of $v$. Each online vertex $u$ draws an independent variable $x_u \in [0,1]$ uniformly at random. We refer to $x_u$ the arrival time of $u$. The time flows continuously from $0$ to $1$. Upon the arrival of vertex $u$, it matches the unmatched neighbor $v$ (if exists) with maximum $w_v \cdot \left(1-f(x_u,y_v) \right)$.

We shall focus on functions $f$ that satisfies the following three conditions:
1) $f$ is non-increasing in the first dimension;
2) $f$ is non-decreasing in the second dimension;
3) $f(x,1) = 1$ and $f(1,y)=0$ for every $x \in [0,1]$ and $y \in [0,1)$.
We use $\mathcal{F}$ to denote the set of all such functions. We remark that for unweighted graphs, any choice of the function $f$ that is strictly monotone in the second dimension would lead to the same classical Ranking algorithm by Karp et al.~\cite{stoc/KarpVV90}.

\paragraph{Primal-dual Analysis.} We follow the (randomized) primal dual analysis of Devanur, Jain, and Kleinberg~\cite{soda/DevanurJK13}. Consider the following linear program of the vertex-weighted bipartite matching problem and its dual.
\begin{align*}
	\max: \quad & \sum_{(u,v)\in E} w_v\cdot s_{uv} && \qquad\qquad & \min: \quad & \sum_{v \in L} t_v + \sum_{u \in R} t_u\\
	\text{s.t.} \quad & \sum_{u:(u,v)\in E} s_{uv} \leq 1 && \forall v\in L & \text{s.t.} \quad & t_u + t_v \geq w_v && \forall (u,v)\in E \\
	& \sum_{v:(u,v)\in E} s_{uv} \leq 1 && \forall u\in R & & t_v \geq 0 && \forall v \in L \\
	& s_{uv} \geq 0 && \forall (u,v)\in E & & t_u \geq 0 && \forall u \in R
\end{align*}
The Ranking algorithm induces a natural assignment to the primal and dual variables for any $\vect{x} \in [0,1]^{R}$ and $\vect{y} \in [0,1]^{L}$. 
Whenever an edge $(u,v)$ is matched, let $s_{uv} = 1$, $t_u = (1-f(x_u,y_v)) \cdot w_v$, and $t_v = f(x_u,y_v) \cdot w_v$.
By Lemma 2.1 of Huang et al.~\cite{talg/HuangTWZ19}, to establish a competitive ratio of $\Gamma$, it suffices to prove that
\[
\underset{\vect{x},\vect{y}}{\E}\left[t_u + t_v \right] \ge \Gamma \cdot w_v~.
\]

\paragraph{Economic Interpretation.} The Ranking algorithm and the primal-dual analysis admit a folklore economic interpretation, refer to e.g. \cite{sosa/EdenFFS21}. Consider the offline vertices as goods and the online vertices as buyers. Then, $w_v \cdot f(x,y_v)$ is the price of good $v$ at time $x$. When buyer $u$ arrives at time $x_u$, she buys the good that gives her the largest utility, i.e., $w_v \cdot (1-f(x_u,y_v))$. In this language, each dual variable $t_v$ for $v\in L$ corresponds to the price of selling of good $v$; and each dual variable $t_u$ for $u \in R$ correspond to the utility of buyer $u$.

\begin{lemma}[Lemma 3.1 of Huang et al.~\cite{talg/HuangTWZ19}. Refer to Figure~\ref{fig:matching_status}]
\label{lem:alpha_beta}
Fix an arbitrary edge $(u,v) \in E$ and arbitrary $\vect{x}_{\text{-}u} \in [0,1]^{R-u}, \vect{y}_{\text{-}v} \in [0,1]^{L-v}$.
For each $x \in [0,1]$, there exist thresholds $0 \le \beta(x) \le \alpha(x) \le 1$, such that when $u$ arrives at time $x_u =x$:
\begin{itemize}
\item if $y_v < \beta(x)$, $v$ is matched when $u$ arrives;
\item if $y_v \in (\beta(x), \alpha(x))$, edge $(u,v)$ is matched;
\item if $y_v > \alpha(x)$, $v$ is unmatched after $u$'s arrival.
\end{itemize}
Moreover, $\beta(x)$ is a non-decreasing function. 
\end{lemma}

\begin{figure}[h]
\centering
        \begin{subfigure}[b]{0.33\textwidth}
                \centering
                \includegraphics[width=\linewidth]{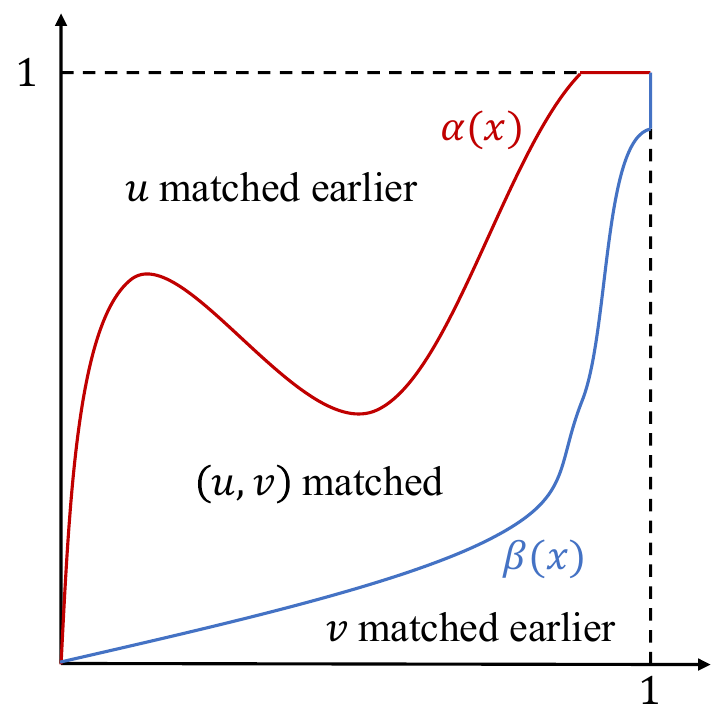}
                \caption{Function $\alpha,\beta$}
                \label{fig:matching_status}
        \end{subfigure}\hfill
        \begin{subfigure}[b]{0.33\textwidth}
                \centering
                \includegraphics[width=\linewidth]{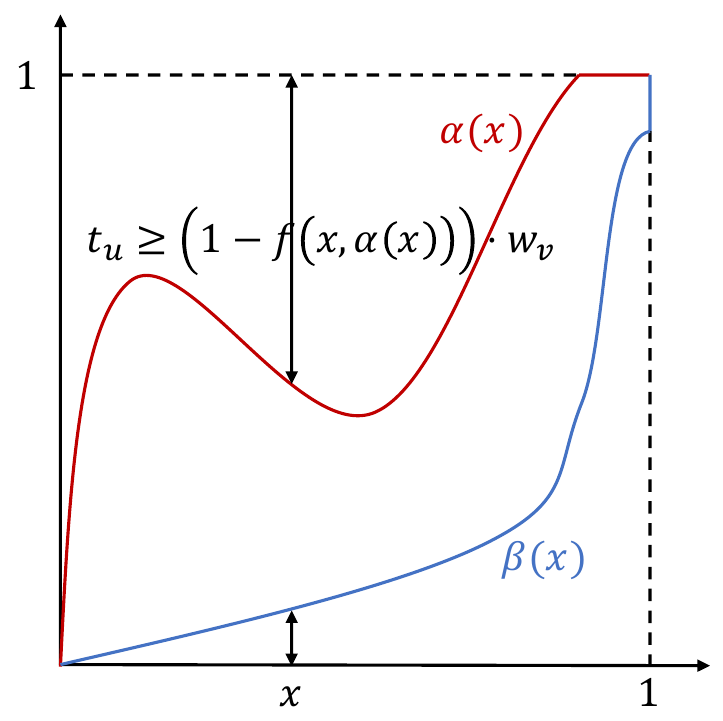}
                \caption{$u$ is matched to $z \ne v$}
                \label{fig:u_gain}
        \end{subfigure}\hfill
        \begin{subfigure}[b]{0.33\textwidth}
                \centering
                \includegraphics[width=\linewidth]{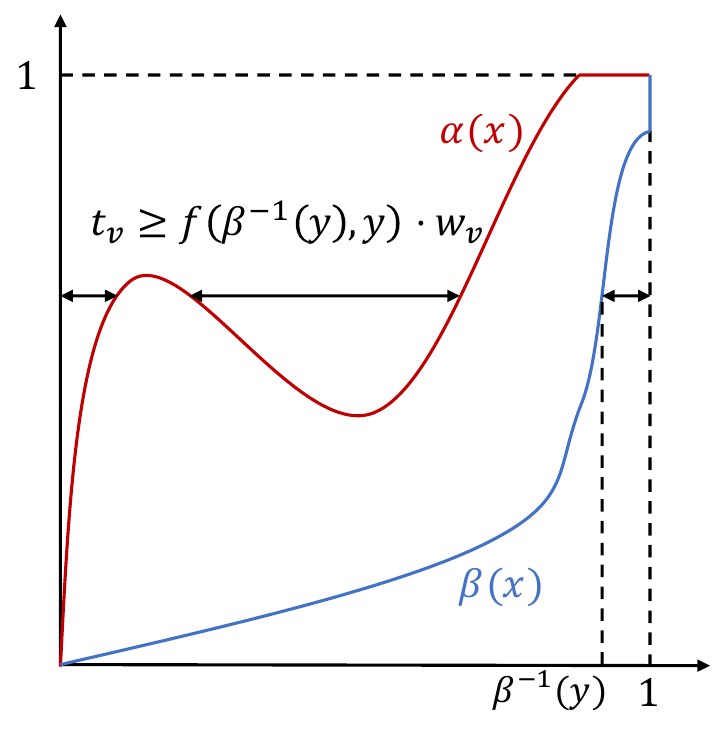}
                \caption{$v$ is matched to $z \ne u$}
                \label{fig:v_gain}
        \end{subfigure}\hfill
        \caption{Illustration of Lemma~\ref{lem:alpha_beta} and \ref{lem:optimization}}
        \label{fig:optimization}
\end{figure}

\subsection{Formulating the Optimization Problem}
We establish the following lower bound on the expected sum of the dual variables $t_u+t_v$ over the randomness of $x_u,y_v$. The bound is suggested by Huang et al.~\cite{talg/HuangTWZ19} in the discussion section of their work (with different notations). Indeed, it is the most natural bound that utilizes all combinatorial properties of our current understanding of the Ranking algorithm in the online bipartite matching problem with random arrivals.
However, both works of Huang et al.~\cite{talg/HuangTWZ19} and Jin and Williamson~\cite{wine/JinW21} are only able to analyze a relaxed version of it and a $0.6688$ barrier is established for the relaxed optimization by Jin and Williamson. 
For completeness, we provide a proof of the lemma. Refer to Figure~\ref{fig:optimization}~.

\begin{lemma}
\label{lem:optimization}
For any edge $(u,v) \in E$ and any $\vect{x}_{\text{-}u} \in [0,1]^{R-u}, \vect{y}_{\text{-}v} \in [0,1]^{L-v}$, let $\alpha(x),\beta(x)$ be defined as in Lemma~\ref{lem:alpha_beta} and $\beta^{-1}(y)$ be the inverse function of $\beta(x)$, i.e., $\beta^{-1}(y) \eqdef \sup\{x: \beta(x) \le y\}$. 
Then we have $\underset{x_u,y_v}{\E}\left[t_u + t_v \right] \ge \Gamma(f,\alpha,\beta) \cdot w_v$, where
\begin{align*}
\Gamma(f,\alpha,\beta) & \eqdef \int_0^1 \alpha(x) \:\dd x - \int_0^1 \beta(x) \:\dd x && ((u,v) \text{ matched})\\
& + \int_0^1 \left( 1-\alpha(x) + \beta(x)\right) \cdot \left( 1 - f(x,\alpha(x)) \right) \:\dd x && (\E[t_u], \text{Figure~\ref{fig:u_gain}})\\
& + \int_0^1 \left( 1-\beta^{-1}(y) \right) \cdot f(\beta^{-1}(y),y) \:\dd y + \int_0^1 \int_{\alpha(x)}^1 f(\beta^{-1}(y),y) \:\dd y \:\dd x && (\E[t_v], \text{Figure~\ref{fig:v_gain}})
\end{align*}
\end{lemma}

\begin{proof}
Notice that for any $x_u \in [0,1]$, if $y_v \in (\beta(x_u),\alpha(x_u))$, $u,v$ are matched to each other and $t_u+t_v =w_v$. Hence we have 
$$
\underset{x_u,y_v}{\E}\left[(t_u + t_v) \cdot \mathbbm{1} \left[ \beta(x_u) < y_v < \alpha(x_u)\right] \right] = \left(\int_0^1 \alpha(x) \:\dd x - \int_0^1 \beta(x) \:\dd x \right) \cdot w_v.
$$
Next, we consider the gain of $t_u$ when $u$ is matched to $z \ne v$. Refer to Figure~\ref{fig:u_gain}. For any $x_u$, $t_u \ge (1-f(x_u,\alpha(x_u))) \cdot w_v$ holds trivially for the special case when $\alpha(x_u)=1$, due to the boundary condition that $f(x,1)=1$. 
Otherwise, when $y_v = \alpha(x_u) + \varepsilon$, $v$ remains unmatched after $u$'s arrival, which implies that $t_u \ge (1-f(x_u,\alpha(x_u))) \cdot w_v$. The same inequality then holds for every $y_v > \alpha(x_u)$. For $y_v < \beta(x_u)$, by the Monotonicity Lemma (refer to Lemma 2.3 of~\cite{soda/DevanurJK13}), $t_u$ can only be larger in $G=(L,R)$ than that in $G=(L \setminus \{v\},R)$. Thus, we still have $t_u \ge (1-f(x_u,\alpha(x_u))) \cdot w_v$. To sum up,
$$
\underset{x_u,y_v}{\E}\left[t_u  \cdot \mathbbm{1}\left[ y_v < \beta(x_u)  \text{ or } y_v > \alpha(x_u)\right] \right] \ge \int_0^1 \left( 1-\alpha(x) + \beta(x)\right) \cdot \left( 1 - f(x,\alpha(x)) \right) \:\dd x \cdot w_v.
$$
Similarly, we have that $t_v \ge f(\beta^{-1}(y_v),y_v) \cdot w_v$ for $y_v > \alpha(x_u)$ or $y_v < \beta(x_u)$. Refer to Figure~\ref{fig:v_gain}. But we integrate the two parts from different directions.
\begin{align*}
& \underset{x_u,y_v}{\E}\left[t_v  \cdot \mathbbm{1}\left[ y_v < \beta(x_u) \right] \right] \ge \int_0^1 \left( 1- \beta^{-1}(y)\right) \cdot f(\beta^{-1}(y), y) \:\dd y \cdot w_v \\
& \underset{x_u,y_v}{\E}\left[t_v  \cdot \mathbbm{1}\left[ y_v > \alpha(x_u) \right] \right] \ge \int_0^1 \int_{\alpha(x)}^1 f(\beta^{-1}(y),y) \:\dd y \:\dd x \cdot w_v 
\end{align*}
This concludes the proof of the lemma.
\end{proof}

\begin{remark}
For unweighted graphs, it can be shown that $\alpha(x)$ must also be non-decreasing and we can define its inverse function as $\alpha^{-1}(y) \eqdef \sup\{x:\alpha(x) \le y\}$. This would allow us to simplify the third line of the equation above to $\int_0^1 (1-\beta^{-1}(y)+\alpha^{-1}(y))\cdot f(\beta^{-1}(y),y) \;\dd y$, so that $\Gamma(f,\alpha,\beta)$ becomes symmetric for $\alpha(x)$ and $\beta^{-1}(y)$. However, for vertex-weighted graphs, $\alpha(x)$ is not necessarily monotone (refer to Remark 3.2 of~\cite{talg/HuangTWZ19}). Our lower bound result in Section~\ref{subsec:lower} shall treat the two curves $\alpha,\beta$ in an asymmetric way and our upper bound result in Section~\ref{subsec:upper} shall treat the two curves in a symmetric way. 
\end{remark}

For an arbitrary $f \in \mathcal{F}$, let $S_f$ denote the set of all possible pairs $(\alpha, \beta)$ that admit a graph $G$, an edge $(u,v)$ and $\vect{x}_{\text{-}u}, \vect{y}_{\text{-}v}$ that implements the two curves. Hereafter, we focus on solving the following optimization problem. 
\begin{equation}
\label{eq:opt}
\textbf{Optimization:}\quad \quad \max_{f \in \mathcal{F}} \inf_{(\alpha,\beta) \in S_f} \ \Gamma(f,\alpha,\beta)~.
\end{equation}
Our main result is the following.
\begin{theorem}
\label{thm:opt}
	The value of optimization \eqref{eq:opt} is between $0.6862$ and $0.7027$.
\end{theorem}

%% file: discretization.tex
Despite the closed-form formulation of the max-min optimization problem, it is unclear how to solve it even numerically. Observe that the inner optimization itself is essentially a variational problem, i.e., finding $\alpha,\beta$'s to minimize $\Gamma(f,\alpha,\beta)$ against a fixed $f$. We would like to remark that the previous works by Huang et al.~\cite{talg/HuangTWZ19} and Jin and Williamson~\cite{wine/JinW21} are aware of this optimization problem but failed to solve it. Indeed, both works state as an interesting question to solve optimization~\eqref{eq:opt}.
In this section, we establish two family of linear programs that lower and upper bound optimization \eqref{eq:opt} respectively. A key observation is that for any fixed $\alpha$ and $\beta$, $\Gamma(f,\alpha,\beta)$ is linear as a functional of $f$. This observation is crucial for constructing our linear programs.

\subsection{Lower Bound}
\label{subsec:lower}
We first restrict ourselves to discretized functions $f$ in the current form.
Given an $g: \{0,1,\ldots,m \} \times \{0,1,\ldots,n\} \to [0,1]$ that is 1) non-increasing in the first dimension, 2) non-decreasing in the second dimension, and 3) $g(i,n) = 1$, $g(m,j)=0$ for $i \in \{0,1,\ldots,m\}$ and $j \in \{0,1,\ldots,n-1\}$, let $f$ be defined as $f(x,y) = g(\lfloor mx\rfloor, \lfloor ny\rfloor)$. Refer to Figure~\ref{fig:discrete_f}.
\begin{figure}[h]
\centering
\begin{subfigure}[b]{0.3\textwidth}
\centering
\includegraphics[width=\linewidth]{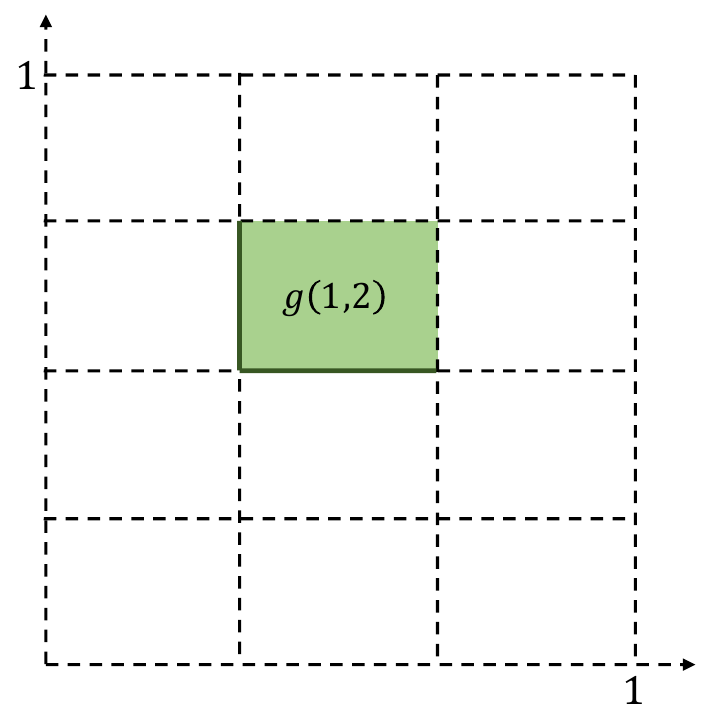}
\caption{For $x\in [\frac{1}{3},\frac{2}{3}), y \in [\frac{1}{2},\frac{3}{4})$, let $f(x,y)=g(1,2)$}
\label{fig:discrete_f}
\end{subfigure}\hfill
\begin{subfigure}[b]{0.3\textwidth}
\centering
\includegraphics[width=\linewidth]{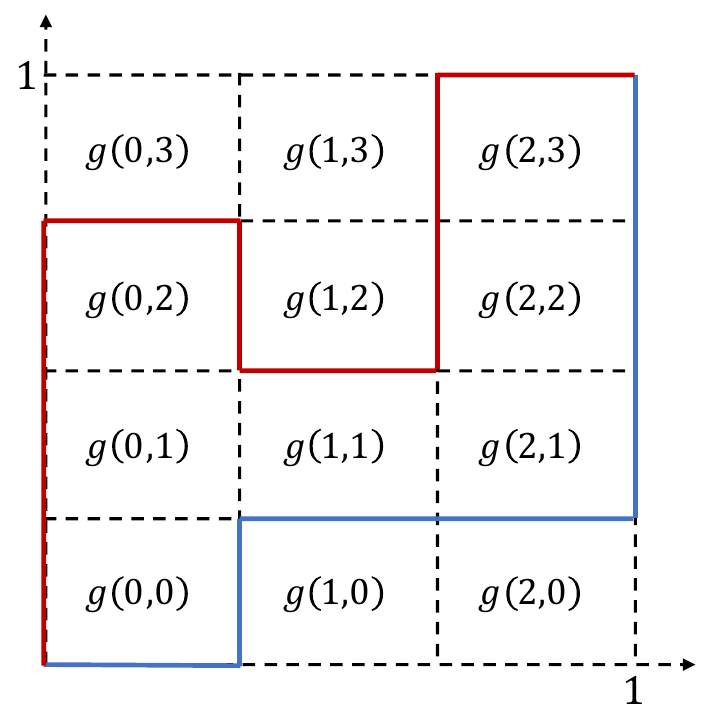}
\caption{$g(*,4)=1$ and $g(3,*)=0$; $g(3,4)$ is not used in our analysis}
\label{fig:discrete_ab}
\end{subfigure}\hfill
\begin{subfigure}[b]{0.36\textwidth}
\centering
\includegraphics[width=\linewidth]{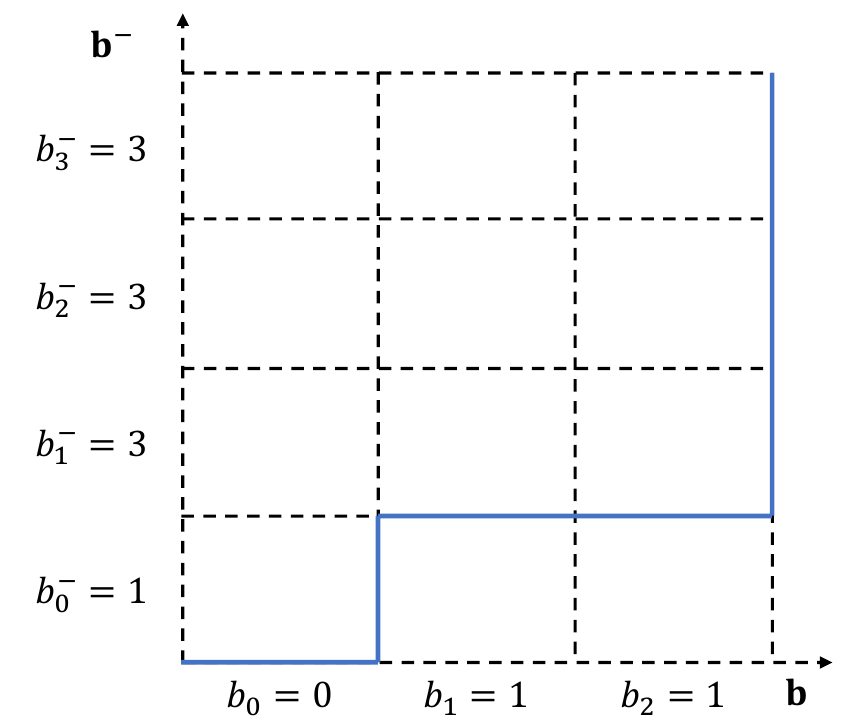}
\caption{Illustration of $\vect{b}^-,\vect{b}$, where $b_3=4$ and $b^-_4=3$}
\label{fig:discrete_beta}
\end{subfigure}
\caption{Illustration of the discretization for $m=3,n=4$}
\end{figure}

Our first lemma states that for such discretized functions $f$, it suffices to study discretized functions $\alpha, \beta$. More formally, we defined grid paths on the $m \times n$ grid.
\begin{definition}
A function $\theta:[0,1] \to [0,1]$ is called a grid path on the $m \times n$ grid if
\begin{itemize}
\item $\theta(x) = \theta(x')$, $\forall  x,x' \in [0,1]$ s.t. $\lfloor mx\rfloor =\lfloor mx' \rfloor$;
\item $\theta(x) \in \left\{0, \frac{1}{n}, \frac{2}{n} \ldots, 1 \right\}, \forall x \in [0,1]$;
\item $\theta(1) = 1$.
\end{itemize}
\end{definition}

\begin{lemma}
\label{lem:grid_lower}
For the family of discretized functions $f$ above, without loss of generality, we can restrict ourselves to grid paths $\alpha, \beta$ on the $m\times n$ grid, where $\beta$ is further monotonic (refer to Figure~\ref{fig:discrete_ab}).
\end{lemma}
\begin{proof}
By discretizing the $x$-axis, we are effectively studying the independent random arrival model:
\begin{itemize}
\item Let there be $m$ stages (from early to late).
\item Each online vertex $u$ draws independently $j_u \in \{1,2,\ldots,m\}$ uniformly at random and arrives at the $j_u$-th stage. Then, $x_u \in \left[ \frac{j-1}{n}, \frac{j}{m} \right)$ corresponds to case when $u$ arrives at the $j$-th stage.

\item If there are multiple vertices arriving at the same stage $j$, order them arbitrarily but in a consistent manner\footnote{Let $R_j$ be the set of online vertices that arrive in the $j$-th stage. They might arrive in an arbitrary order as long as it satisfies the following consistency condition. If a vertex $u \in R_j$ is switched to another stage, the relative arrival order of the remaining vertices in $R_j$ does not change.}. 
\end{itemize}
Observe that when $m=1$,  we are effectively studying the online bipartite matching problem with the worst-case arrival order; when $m \to \infty$, we are effectively studying the random arrival order model since with $o(1)$ probability, two vertices would arrive at the same stage. Thus, $\alpha(x)=\alpha(x')$ (and $\beta(x)=\beta(x')$) for every $x,x' \in \left[\frac{j}{m},\frac{j+1}{m}\right), j \in \{0,\ldots,m-1\}$, since $x,x'$ induces the same arrival order.

By discretizing the $y$-axis, we are effectively running a discretized version of Ranking where each offline vertices have only $n$ different ranks rather than a continuous rank between $[0,1]$. If $y,y' \in \left[\frac{i}{n},\frac{i+1}{n}\right), i \in \{0,\ldots,n-1\}$, the corresponding prize of $v$ at any time $x$ are the same, i.e., $w_v \cdot (1-f(x,y)) = w_v \cdot (1-f(x,y'))$. Hence, they would lead to the same matching by Ranking. According to our definition of $\alpha(x)$ (and $\beta(x)$), the matching status is changed when $y_v$ goes from $\alpha^-(x)$ to $\alpha^+(x)$ (resp., from $\beta^-(x)$ to $\beta^+(x)$). It must be the case that $\alpha(x)$ (and $\beta(x)$) take values of $\frac{i}{n}$ for $i \in \{0,1,\ldots,n\}$. 

The two observations together lead to the conclusion that $\alpha(x),\beta(x)$ are both grid paths on the $m \times n$ grid.
\end{proof}

We introduce an alternative discretized representation of a grid path that is more convenient for our further analysis. Refer to Figure~\ref{fig:discrete_beta} for the following observation and definition.
\begin{observation}
\label{obs:grid_vector}
Every grid path $\beta$ on the $m \times n$ grid can be represented by a vector $\vect{b} \in \{0,1,\ldots,n\}^{m+1}$ that are indexed from $0$ to $m$, where $b_i = n \cdot \beta\left(\frac{i}{m}\right)$. 
\end{observation}

\begin{definition}
We use $\mathcal{B}$ to denote the set of non-decreasing vectors $\vect{b} \in \{0,1,\ldots,n \}^{m+1}$, i.e., $\mathcal{B} = \{ \vect{b} : b_i \ge b_{i'}, \forall i > i'\}$.
For each $\vect{b} \in \mathcal{B}$, we use $\vect{b}^{-}$ to denote the inverse vector of $\vect{b}$, i.e., $b_j^{-} \eqdef \min\{i : b_i > j\}$.	
\end{definition}

Now, we are ready to establish a linear program that serves as a lower bound of Optimization~\eqref{eq:opt}.

\begin{lemma}
\label{lem:lp_lower}
The value of optimization~\eqref{eq:opt} is at least the value of the following linear program:
\begin{align*}
\max_{\Gamma,g,h}: \quad & \Gamma \tag{LP: Lower Bound}\\
\normalfont \text{s.t.}: \quad &  \Gamma \le \frac{1}{n} \cdot \sum_{j=0}^{n-1} \left( 1- \frac{b^{-}_{j}}{m} \right) \cdot g(b^{-}_{j},j) - \frac{1}{m} \cdot \sum_{i=0}^{m-1} \frac{b_i}{n} + \frac{1}{m} \cdot \sum_{i=0}^{m-1} h(i,\vect{b}) && \forall \vect{b} \in \mathcal{B} \\
& h(i,\vect{b}) \le \frac{j}{n} + \left( 1- \frac{j}{n} + \frac{b_i}{n} \right) \cdot (1-g(i,j)) + \frac{1}{n} \cdot \sum_{k=j}^{n-1} g(b_k^{-},k) && \forall i,\vect{b}, b_i\le j \le n \\
& g(i,j) \le g(i,j+1) && \forall 0 \le i \le m, 0 \le j < n \\
& g(i,j) \ge g(i+1,j) && \forall 0 \le i < m, 0 \le j \le n \\
& g(i,n) = 1 && \forall 0 \le i \le m \\
& g(m,j)=0 && \forall 0 \le j < n
\end{align*}
\end{lemma}
\begin{proof}
We use $g^*,h^*, \Gamma^*$ to denote the optimal solution of the above linear program. And let $f(x,y) \eqdef g^*(\lfloor mx\rfloor, \lfloor ny\rfloor)$. We rearrange the bound of $\Gamma(f,\alpha,\beta)$ from Lemma~\ref{lem:optimization} as the following:
\begin{align*}
\Gamma(f,\alpha,\beta) & = \int_0^1 \left( 1-\beta^{-1}(y) \right) \cdot f(\beta^{-1}(y),y) \:\dd y - \int_0^1 \beta(x) \:\dd x \\
& + \int_0^1 \left( \alpha(x)+\left( 1-\alpha(x) + \beta(x)\right) \cdot \left( 1 - f(x,\alpha(x)) \right)  + \int_{\alpha(x)}^1 f(\beta^{-1}(y),y) \:\dd y \right)\:\dd x 
\end{align*}
By Lemma~\ref{lem:grid_lower}, it is suffices to study grid paths $\alpha, \beta$ and we use $\vect{a}, \vect{b} \in \{0,1,\ldots,n\}^{m+1}$ to denote the corresponding vector of $\alpha,\beta$ respectively.
Then we have that
\begin{align*}
& \int_0^1 \left( 1-\beta^{-1}(y) \right) \cdot f(\beta^{-1}(y),y) \:\dd y - \int_0^1 \beta(x) \:\dd x = \frac{1}{n} \cdot \sum_{j=0}^{n-1} \left(1-\frac{b^{-}_j}{m} \right) \cdot g^*(b^-_j,j) - \frac{1}{m} \sum_{i=0}^{m-1} \frac{b_i}{n}~; \\
\text{and} \quad &  \int_0^1 \alpha(x) \:\dd x + \int_0^1 \left( 1-\alpha(x) + \beta(x)\right) \cdot \left( 1 - f(x,\alpha(x)) \right) \:\dd x + \int_0^1 \int_{\alpha(x)}^1 f(\beta^{-1}(y),y) \:\dd y \:\dd x \\
& = \frac{1}{m} \sum_{i=0}^{m-1} \left(\frac{a_i}{n}+ \left( 1-\frac{a_i}{n}+\frac{b_i}{n} \right) \cdot (1-g^*(i,a_i))  + \frac{1}{n} \cdot \sum_{k=a_i}^{n-1} g^*(b^-_k,k)\right) \\
& \ge \frac{1}{m} \sum_{i=0}^{m-1} \min_{j \ge b_i}\left( \frac{j}{n} + \left( 1-\frac{j}{n}+\frac{b_i}{n} \right) \cdot (1-g^*(i,j))  + \frac{1}{n} \cdot \sum_{k=j}^{n-1} g^*(b^-_k,k)\right) \ge \frac{1}{m} \cdot \sum_{i=0}^{m-1} h^*(i,\vect{b}),
\end{align*}
where the first inequality follows from the fact that $a_i \in \{b_i,\ldots,n\}$ and the last inequality follows from the second family of constraints of our linear program.
Putting the two bounds together concludes that $\Gamma(f,\alpha,\beta) \ge \Gamma^*$ for all $\alpha,\beta$.
\end{proof}

%% file: upper.tex
\subsection{Upper Bound}
\label{subsec:upper}

In this subsection, we establish a linear program that upper bounds optimization~\eqref{eq:opt}. We shall again focus on grid paths $\alpha,\beta$. We first show that all such grid paths are implementable.
\begin{figure}[h]
\centering
\includegraphics[width=0.6\linewidth]{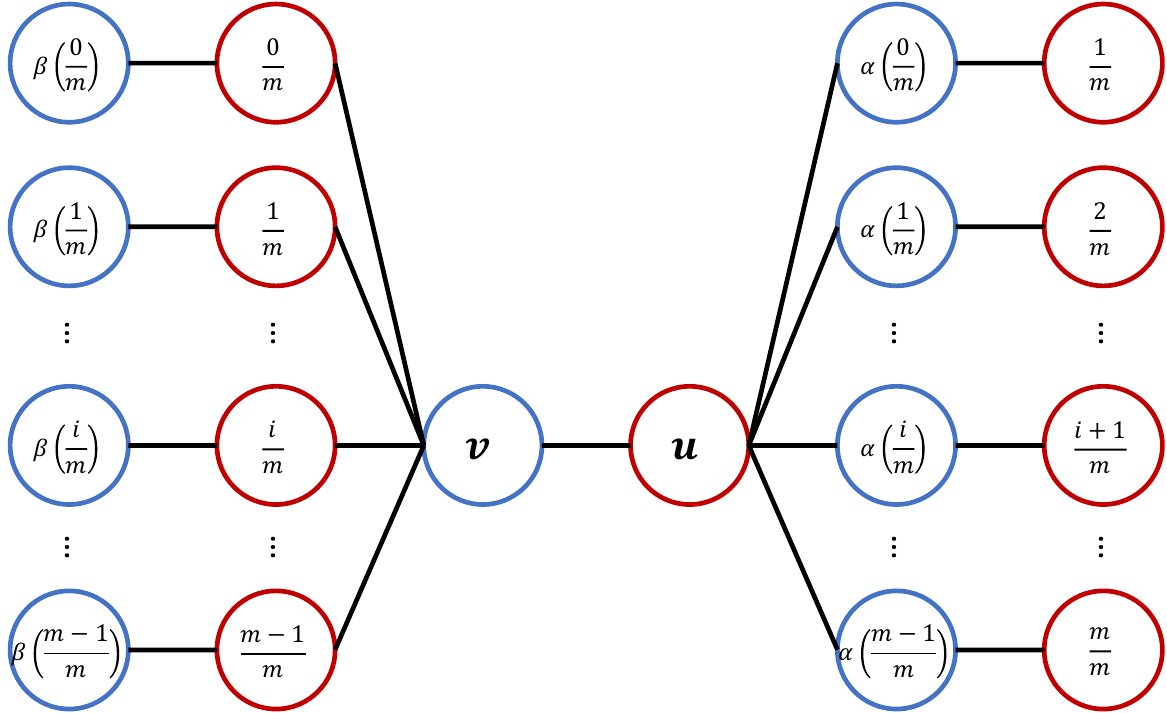}
\caption{The red vertices correspond to the online vertices and the numbers written in the circles correspond to their arrival time. The blue vertices correspond to the offline vertices and the number written in the circles correspond to their ranks.}
\label{fig:graph_a_b}
\end{figure}

\begin{lemma}
\label{lem:upper}
For any non-decreasing grid paths $\alpha,\beta$ on the $m \times n$ grid with $\beta(x) \le \alpha(x)$, there exists an unweighted graph $G$, an edge $e=(u,v) \in G$ and $\vect{x}_{\text{-}u},\vect{y}_{\text{-}v}$ that implements $\alpha(x),\beta(x)$. I.e., $\alpha,\beta$ satisfy the conditions as stated in Lemma~\ref{lem:alpha_beta}.
\end{lemma}
\begin{proof}
Consider a graph with $4m+2$ vertices including $u,v$ and $4m+1$ edges (refer to Figure~\ref{fig:graph_a_b}). The red vertices correspond to the online vertices and the numbers therein indicate their arrival times. The blue vertices correspond the offline vertices and the numbers therein indicate their ranks. Suppose that the vertex $u$ arrives at time $x_u \in \left( \frac{i}{m}, \frac{i+1}{m} \right) $.
Then if $y_v <\beta\left(\frac{i}{m}\right)$, $v$ must be matched before time $\frac{i}{m}$; otherwise, $u$'s best choice is between $y_v$ and $\alpha\left(\frac{i}{m}\right)$. That is, if $y_v \in \left(\beta\left(\frac{i}{m}\right), \alpha\left(\frac{i}{m}\right)\right)$, $(u,v)$ is matched and otherwise $u$ matches to a neighbor with rank $\alpha\left(\frac{i}{m}\right)$.
\end{proof}
\begin{remark}
This lemma justifies the optimality of Optimization~\eqref{eq:opt} under our current knowledge of the Ranking algorithm and the randomized primal-dual framework. Indeed, if we restrict the analysis to set the dual variables as we do, i.e., $t_u=1-f(x_u,y_v)$ and $t_v=f(x_u,y_v)$ when $(u,v)$ is matched by the algorithm, the above lemma gives a graph for which $\E[t_u+t_v]=\Gamma(f,\alpha,\beta)$.
\end{remark}

\begin{lemma}
\label{lem:lp_upper}
The value of optimization~\eqref{eq:opt} is at most the value of the following linear program:
\begin{align*}
\max_{\Gamma,g}: \quad & \Gamma \tag{LP: Upper Bound}\\
\normalfont \text{s.t.}: \quad &  \Gamma \le \frac{1}{m} \cdot \sum_{i=0}^{m-1} \left(\frac{a_i}{n}  -  \frac{b_i}{n}\right) + \frac{1}{m} \cdot \sum_{i=0}^{m-1} \left(1-\frac{a_i}{n}+\frac{b_i}{n}\right)\left(1-g(i+1,a_i)\right) \\
& \phantom{ \Gamma} + \frac{1}{n} \cdot \sum_{j=0}^{n-1} \left(1-\frac{b^{-}_{j}}{m}\right) \cdot g(b^{-}_{j},j+1) +\frac{1}{m} \cdot \sum_{i=0}^{m-1} \left( \frac{1}{n} \cdot \sum_{j=a_i}^{n-1} g(b^{-}_{j},j+1) \right) && \forall \vect{a}, \vect{b} \in \mathcal{B} \text{ s.t. } a_i \ge b_i \\
& g(i,j) \le g(i,j+1) && \forall 0 \le i \le m, 0 \le j < n \\
& g(i,j) \ge g(i+1,j) && \forall 0 \le i < m, 0 \le j \le n \\
& g(i,n) = 1 && \forall 0 \le i \le m \\
& g(m,j)=0 && \forall 0 \le j < n
\end{align*}
\end{lemma}

\begin{proof}
For an arbitrary function $f \in \mathcal{F}$, we use $g:\{0,1,\ldots,m\} \times \{0,1,\ldots,n\} \to [0,1]$ to denote the function values of $f$ on grid points, i.e., $g(i,j) \eqdef f\left( \frac{i}{m}, \frac{j}{n}\right)$.
According to Lemma~\ref{lem:upper}, all grid paths $\alpha,\beta$ on the $m\times n$ grid are implementable. In other words, we have that 
\[
\Gamma^* \eqdef \inf_{(\alpha,\beta) \in S_f} \Gamma(f,\alpha,\beta) \le \min_{\text{grid } \alpha, \beta} \Gamma(f, \alpha,\beta)~.
\]
It suffices to show that $\Gamma^*$ and $g$ are feasible to the linear program as stated in the lemma.
For any non-decreasing grid paths $\alpha,\beta$, we use $\vect{a}, \vect{b} \in \{0,1,\ldots,n\}^{m+1}$ to denote their corresponding vectors as in Observation~\ref{obs:grid_vector}. Then we have,
\begin{align*}
\Gamma(f,\alpha,\beta)  = & \int_0^1 \alpha(x) \:\dd x - \int_0^1 \beta(x) \:\dd x  + \int_0^1 \left( 1-\alpha(x) + \beta(x)\right) \cdot \left( 1 - f(x,\alpha(x))\right) \:\dd x \\
& + \int_0^1 \left( 1-\beta^{-1}(y) \right) \cdot f(\beta^{-1}(y),y) \:\dd y + \int_0^1 \int_{\alpha(x)}^1 f(\beta^{-1}(y),y) \:\dd y \:\dd x \\
= & \sum_{i=0}^{m-1} \int_{\frac{i}{m}}^{\frac{i+1}{m}} \left(\alpha(x) - \beta(x)\right) \:\dd x  + \sum_{i=0}^{m-1} \int_{\frac{i}{m}}^{\frac{i+1}{m}} \left( 1-\alpha(x) + \beta(x)\right) \cdot \left( 1 - f(x,\alpha(x))\right) \:\dd x \\
& + \sum_{j=0}^{n-1} \int_{\frac{j}{n}}^{\frac{j+1}{n}} \left( 1-\beta^{-1}(y) \right) \cdot f(\beta^{-1}(y),y) \:\dd y + \sum_{i=0}^{m-1} \int_{\frac{i}{m}}^{\frac{i+1}{m}} \sum_{j=a_i}^{n-1} \int_{\frac{j}{n}}^{\frac{j+1}{n}}  f(\beta^{-1}(y),y) \:\dd y \:\dd x \\
\le & \frac{1}{m} \cdot \sum_{i=0}^{m-1} \left(\frac{a_i}{n}  -  \frac{b_i}{n}\right) + \frac{1}{m} \cdot \sum_{i=0}^{m-1} \left(1-\frac{a_i}{n}+\frac{b_i}{n}\right)\left(1-g(i+1,a_i)\right) \\
& + \frac{1}{n} \cdot \sum_{j=0}^{n-1} \left(1-\frac{b^{-}_{j}}{m}\right)g(b^{-}_{j},j+1) +\frac{1}{m} \cdot \sum_{i=0}^{m-1} \left( \frac{1}{n} \cdot \sum_{j=a_i}^{n-1} g(b^{-}_{j},j+1) \right),
\end{align*}
where the inequality holds by the fact that $f(x,\alpha(x)) \ge f \left(\frac{i+1}{m}, \alpha\left(\frac{i}{m}\right) \right) = g(i+1,a_i)$  for $x \in \left(\frac{i}{m},\frac{i+1}{m}\right)$ and that $f(\beta^{-1}(y),y) \le f \left(\beta^{-1}\left(\frac{j}{m}\right), \frac{j+1}{m}\right) = g(b_j^-,j+1)$  for $y \in \left(\frac{j}{m},\frac{j+1}{m}\right)$. Here, we use the property that $f$ is non-increasing in the first dimension and non-decreasing in the second dimension. This concludes the proof of the lemma.
\end{proof}

%% file: numerical.tex
In this section, we discuss our numerical results for solving the two linear programs established in Lemma~\ref{lem:lp_lower} and Lemma~\ref{lem:lp_upper} and their implications. Notice that although we have transferred the continuous max-min optimization problem~\eqref{eq:opt} to linear programs, our linear programs (LP: Lower Bound) and (LP: Upper Bound) consist of exponential number of constraints. All the numerical experiments below are conducted using Gurobi solver on a personal laptop\footnote{Our code is available at \url{https://anonymous.4open.science/r/Ranking_alg-52D9}}.

\subsection{Lower Bound}
\label{subsec:lower_ratio}
We first study the program (LP: Lower Bound). The bottleneck is to consider all possible monotone grid paths $\vect{b}$ on the $m \times n$ grid and there are ${m+n \choose m}$ many of such paths.
We would like to remark that there are two interesting regimes regarding the choice of $m$ and $n$: 1) when $m=n$, this is the most natural choice of the discretization, especially for the case when the graph is unweighted, since the two sides of the graph are symmetric; 2) when $m$ is a constant and $n \to \infty$, this regime corresponds to the competitive ratio of Ranking under the independent random arrival model with $m$ stages, as a corollary of our Lemma~\ref{lem:grid_lower}. Our results are summarized in Table~\ref{table:lower}. 

\begin{table}[h]
\centering
\begin{tabular}{|l|l|l|l|}
\hline
$m$ & $n=m$   & $n \to \infty$        & polyLP$'(n)$ of ~\cite{stoc/MahdianY11}  \\ \hline
1   & 0.5     & $1-1/e\approx 0.632120$  & 0.5                		\\ \hline
2   & 0.625   & 0.665640 when $n=240$      & 0.625                	\\ \hline
3   & 0.641723 & 0.676339 when $n=90$      & 0.641723                \\ \hline
4   & 0.657429 &  0.681097 when $n=50$                  & 0.657429        			\\ \hline
5   & 0.667052 &  0.683205 when $n=30$                  & 0.667052        		 	\\ \hline
6   & 0.673323 &  0.683958 when $n=20$                  & 0.673323        			\\ \hline
7   & 0.677328 &  0.684458 when $n=16$                  & 0.677393        			\\ \hline
8   & 0.680347 &  0.685325 when $n=15$                  & 0.680363        			\\ \hline
9   & 0.682680 &  0.685399 when $n=13$                  & 0.682681        			\\ \hline
10  & 0.684397 & 0.685568 when $n=12$                  & 0.684413        			\\ \hline
11  & 0.685694 & \textbf{0.686254} when $n=12$                  & 0.685728       			\\ \hline
\end{tabular}
\caption{Numerical results of (LP: Lower Bound)}
\label{table:lower}
\end{table}

The best lower bound $0.6862$ of Theorem~\ref{thm:opt} is achieved when $m=11,n=12$, the largest regime that we can afford. Observe that there are ${23 \choose 11} = 1352078$ number of different grid paths to be considered, leaving alone the complexity of solving the linear program.

We would like to remark that when $m=1, n\to \infty$, our linear program is able to reproduce the tight competitive analysis of Ranking (that corresponds to $f(x,y)=e^{y-1}$). 
To the best of our knowledge, our results for constant number of $m$ establish the first non-trivial competitive ratio (i.e. beating $1-1/e$) of Ranking beyond the classical random arrival order assumption, even for unweighted graphs. Specifically, we show that $m^{|R|}$ number of different orders suffices to bypass the $1-1/e$ barrier. In terms of entropy, our relaxed version of the random arrival has entropy $|R| \cdot \log m$ while the classical random arrival order has entropy $|R| \cdot \log |R|$.

\begin{theorem*}
When each online vertex arrives independently at early or late stage (i.e., $m=2$), the Ranking algorithm is at least $0.6656$-competitive. 
When each online vertex arrives independently at early, middle, or late stage (i.e., $m=3$), the Ranking algorithm is at least $0.6763$-competitive.	
\end{theorem*}
Notice that the case when $m=2$ is the minimal independent arrival order and our ratio of $0.6656$ has already improved upon the $0.6629$ competitive ratio of Jin and Williamson~\cite{wine/JinW21}; and when $m=3$, our ratio of $0.6763$ surpasses the $0.6688$ barrier of previous approaches, as observed by Jin and Williamson.

Finally, we would like to remark that our numerical results of the $m=n$ regime mysteriously matches the factor-revealing LP of Mahdian and Yan~\cite{stoc/MahdianY11} as summarized in Table~\ref{table:lower}. For $n \le 6$, the numerical results match exactly; and for $7 \le n \le 11$, the gaps are no larger than $10^{-4}$. We shall elaborate more on this observation in the discussion section.

\subsection{Upper Bound}
\label{subsec:num_upper}
Next, we study the program (LP: Upper Bound), which has more constraints than (LP: Lower Bound) as we now need to study all possible grid path pairs $(\vect{a}, \vect{b})$ on the $m \times n$ grid. If we ignore the dominance of $\vect{a}$ over $\vect{b}$ (i.e., $a_i \ge b_i$), there are ${m+n \choose m} \cdot {m+n \choose m}$ possible combinations.
If we stick to the exact solution, the largest regime that we can afford is when $m=n=7$, which lead to an upper bound of $0.7189$.

Nevertheless, in order to establish a valid upper bound, we can choose an arbitrary subset $S$ of all grid path pairs and only apply the corresponding constraints to the original linear program. We use LP($S$) to denote the linear program by restricting the first family of constraints in (LP: Upper Bound) only to pairs $(\vect{a}, \vect{b}) \in S$, and use $\Gamma(S)$ to denote the value of LP($S$).
Our main contribution is an iterative local search algorithm for finding a good set $S$ of small size so that the linear program can be solved efficiently. For simplicity, we restrict ourselves to the case when $m=n$.

\begin{algorithm}[H]\label{heualg}
\SetAlgoLined			
\DontPrintSemicolon		
\SetKwInOut{Input}{\textbf{Input}}		
\SetKwInOut{Output}{\textbf{Output}}	

\Input{a set of grid path pairs $S$}
\Output{a valid upper bound $\Gamma^*$ of (LP: Upper Bound)}
        
Let $\Gamma^* = 1 + \num{1e-9}$. \tcp* {the current best bound}
Solve LP($S$) and let $f_S$ be the corresponding solution. \tcp*{initialization}
\While{ $\Gamma(S) \le \Gamma^{*} - \num{1e-9}$ }{
    Let $\Gamma^* = \Gamma(S)$. \tcp*{update the upper bound}
    \For{each $(\alpha, \beta)$ in $S$}{
        \uIf{$\Gamma(f_S, \alpha, \beta) > \Gamma(S)$}{
            Remove $(\alpha,\beta)$ from $S$.                \tcp*{clean up non-binding constraints}
        }
        \uElse{
            Let $T$ be the set of local perturbations of $(\alpha, \beta)$. \tcp*{local perturbations} 
	    Add every $(\alpha', \beta') \in T$ to $S$ with $\Gamma(f_S, \alpha', \beta') < \Gamma(S) - \num{1e-5}$.
        }
    }
    Solve LP(S) and let $f_S$ be the corresponding solution. \tcp*{solve the new LP($S$)}
}
\Return $\Gamma^{*}$
\caption{A Local Search Algorithm}
\end{algorithm}

At each iteration, we first remove redundant constraints from $S$ so that we are able to keep the size of $S$ small; and then we add grid path pairs to $S$ if the corresponding ratio is smaller than $\Gamma(S)$ by a significant amount. This step can be inefficient if we examine all possible grid path pairs $(\alpha',\beta')$'s. 
Indeed, this step can be viewed as a variational problem for finding the pair $\alpha, \beta$ to minimize $\Gamma(f_S, \alpha, \beta)$ against a given function $f_S$.
Suppose we further fix $\beta$, then the problem becomes a standard calculus of variation. The Euler-Lagrange equations give a system of second-order ordinary differential equations stating that the optimal solution $\alpha^*$ must be locally stationary.

Motivated by this, we restrict our attentions to local perturbations of the grid path pairs from $S$. 
Specifically, for each grid path $\theta$, we define $N(\theta)$ to be the set of grid paths that differs from $\theta$ by just a single unit square. Refer to Figure~\ref{fig:perturbation}. Then, the set of local perturbations $T$ of $(\alpha,\beta)$ is defined as the following and this completes our local search algorithm:
\[
T(\alpha, \beta) \eqdef \{(\alpha', \beta) : \alpha' \in N(\alpha)\} \cup \{(\alpha, \beta') : \beta' \in N(\beta)\}~.
\]
Observe that $|T(\alpha,\beta)| = O(m)$. This makes our algorithm efficient, as long as the size of $S$ is small.

\begin{figure}[h]
\centering
        \begin{subfigure}{0.25\textwidth}
                \centering
                \includegraphics[width=\linewidth]{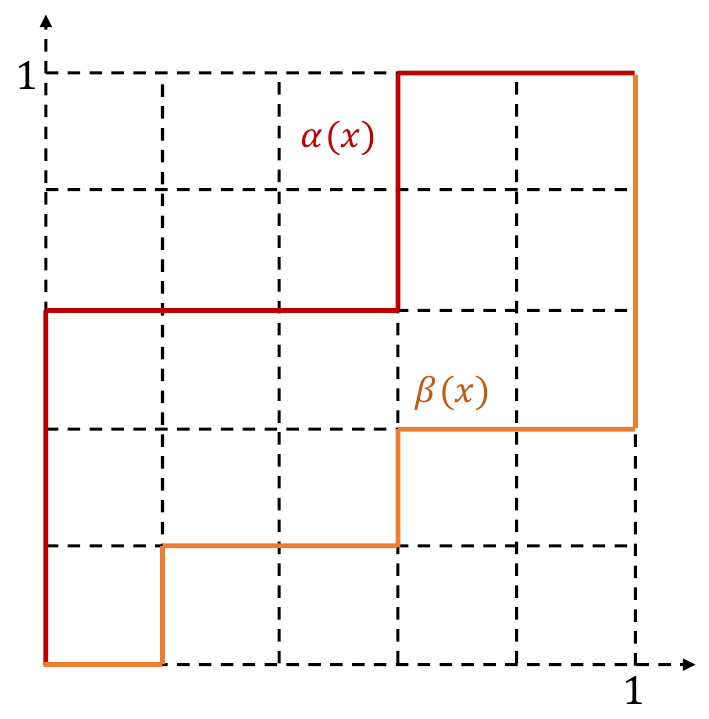}
                \caption{$(\alpha,\beta)$}
        \end{subfigure}\hfill
        \begin{subfigure}{0.25\textwidth}
                \centering
                \includegraphics[width=\linewidth]{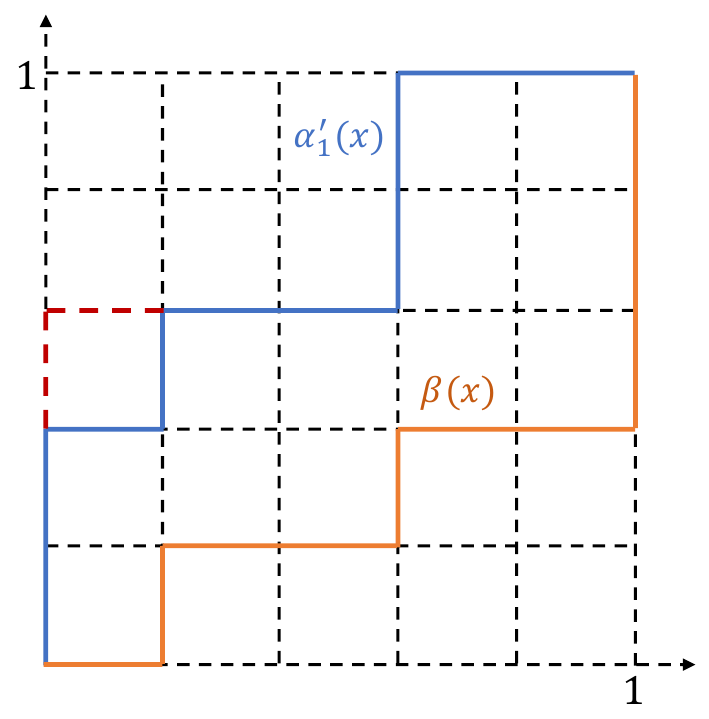}
                \caption{$(\alpha'_1,\beta)$}
        \end{subfigure}\hfill
        \begin{subfigure}{0.25\textwidth}
                \centering
                \includegraphics[width=\linewidth]{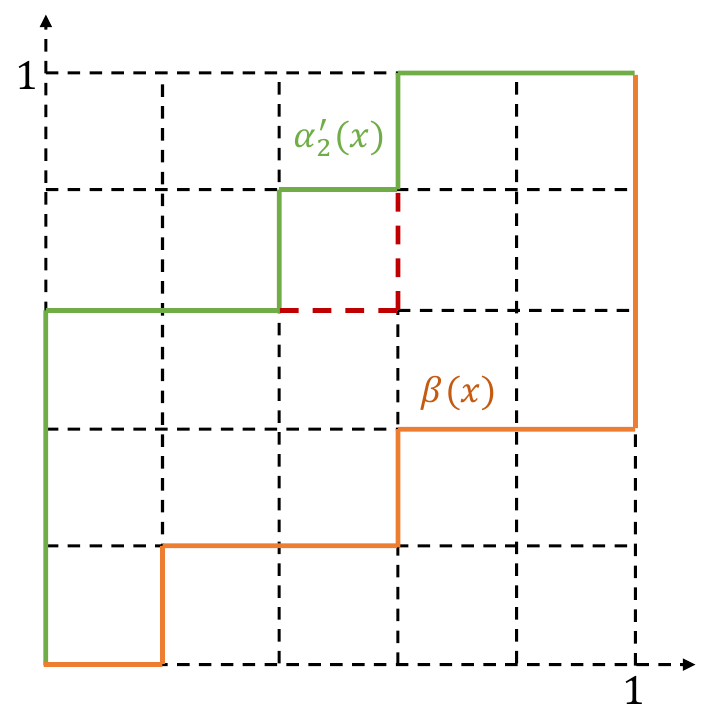}
                \caption{$(\alpha'_2,\beta)$}
        \end{subfigure}\hfill
        \begin{subfigure}{0.25\textwidth}
                \centering
                \includegraphics[width=\linewidth]{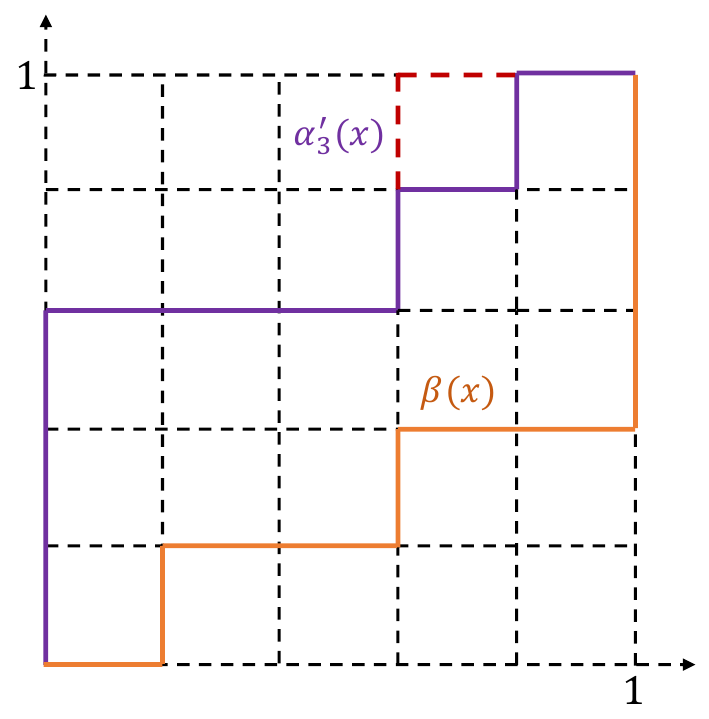}
                \caption{$(\alpha'_3,\beta)$}
        \end{subfigure}\hfill
        \caption{Local Perturbations}
        \label{fig:perturbation}
\end{figure}

\begin{minipage}{\textwidth}
\begin{minipage}[c]{0.49\textwidth}
\centering
\includegraphics[width=\linewidth]{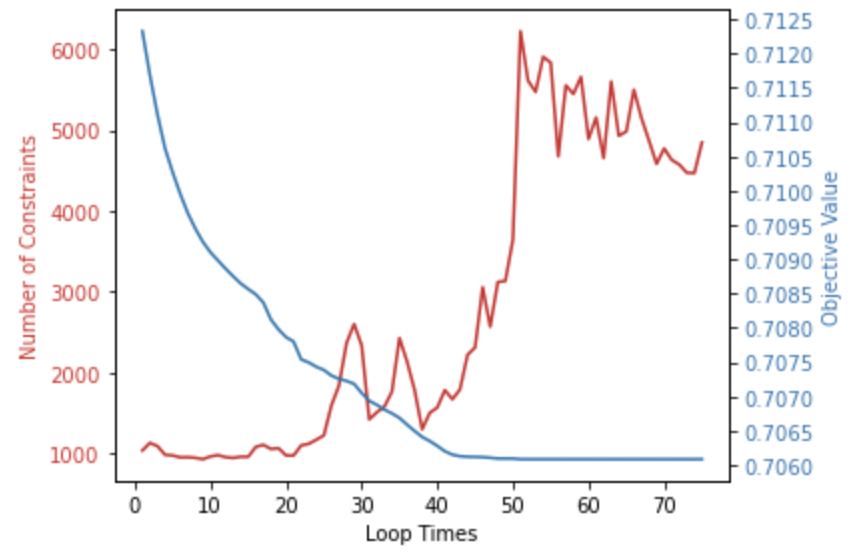}
\captionof{figure}{Local search algorithm for $n=20$}
\label{fig:heualg20}
\end{minipage}
\begin{minipage}{0.49\textwidth}
\centering
\begin{tabular}{|l|l|l|l|}
\hline
$m$ & $m=n$ & polyLP$(n)$ of ~\cite{stoc/MahdianY11}  \\ \hline
1   & 1       & 1                		\\ \hline
2   & 0.75         & 0.75                	\\ \hline
3   & 0.740741       & 0.740741                \\ \hline
4   & 0.733333                   & 0.732456        			\\ \hline
5   & 0.726562                   & 0.725007        		 	\\ \hline
6   & 0.722371                   & 0.720263        			\\ \hline
7   & 0.718931                   & 0.716508        			\\ \hline
10  & 0.713035                  & 0.710998        			\\ \hline
20  & 0.706086                 & 0.704906        			\\ \hline
40  & \textbf{0.702633}                  & 0.701950       			\\ \hline
\end{tabular}
\captionof{table}{Numerical results of (LP: Upper Bound)}
\label{table:upper}
\end{minipage}
\end{minipage}

\vspace{0.1in}
The last missing piece of our approach is the initial choice of the set $S$, which potentially affect the rate of convergence of our algorithm. To this end, we use the set of $S_n$ that we have solved on the $n \times n$ grid as the starting configuration for the $2n \times 2n$ grid. Notice that any grid path on the $n \times n$ grid must also be a grid path on the $2n \times 2n$ grid. This approach proves to be more effective than using randomly generated grid path pairs, according to our numerical experiments.

As an illustration, when $n=20$, we start with the set of binding constraints $S_{10}$ that we have solved for $n=10$. Our local search algorithm converges in about $70$ iterations. As shown in Figure~\ref{fig:heualg20}, the size of $S_{20}$ is less than $6100$ throughout the execution of our algorithm so that each iteration costs less than a minute using Gurobi on a personal laptop. Notice that the original (LP: Upper Bound) have about ${40 \choose 20} \cdot {40 \choose 20} > 10^{22}$ number of constraints, making it impossible to solve exactly.
Our numerical results are summarize in Table~\ref{table:upper}. Here, the stated upper bounds for $m=n \le 7$ are exact solutions of (LP: Upper Bound) and the bounds for $m=n \in \{10,20,40\}$ are achieved by our local search algorithm. The best upper bound $0.7027$ as stated in Theorem~\ref{thm:opt} is achieved when $n=40$. 
A final remark is that the upper bounds we have achieved are also numerically close to the bounds of Mahdian and Yan~\cite{stoc/MahdianY11}.

%% file: discussion.tex
In this work, we revisit the randomized primal-dual analysis of Ranking algorithm for (vertex-weighted) online bipartite matching with random arrivals. We establish upper and lower bounds of the competitive ratio that are numerically close. 
Below, we discuss a few interesting open questions and we hope that our work shed light on characterizing the tight competitive ratio of Ranking for the random arrival model.

\begin{table}[H]
\centering
\begin{tabular}{|l|l|l|l|}
\hline
$n$ & Heuristic algorithm   & LP: Lower Bound      & polyLP$'(n)$ of ~\cite{stoc/MahdianY11}  \\ \hline
1   & 0.5   & 0.5   & 0.5                		\\ \hline
2   & 0.625    & 0.625     & 0.625                	\\ \hline
3   & 0.641723  & 0.641723     & 0.641723                \\ \hline
4   & 0.657429   & 0.657429                & 0.657429        			\\ \hline
5   & 0.667052     & 0.667052              & 0.667052        		 	\\ \hline
6   & 0.673323   & 0.673323                & 0.673323        			\\ \hline
7   & 0.677328      & 0.677328             & 0.677393        			\\ \hline
8   & 0.680347    & 0.680347               & 0.680363        			\\ \hline
9   & 0.682680   & 0.682680                & 0.682681        			\\ \hline
10   & 0.684397   & 0.684397                & 0.684413        			\\ \hline

11   & 0.685720   & 0.685694    & 0.685728                		\\ \hline
12   & 0.686771   &-            & 0.686781                	\\ \hline
13   & 0.687719   &-            & 0.687726                \\ \hline
14   & 0.688533   &-            & 0.688544        			\\ \hline
15   & 0.689275        &-           & 0.689285        		 	\\ \hline
16   & 0.689922        &-           & 0.689931        			\\ \hline
17   & 0.690502       &-            & 0.690511        			\\ \hline
18   & 0.691053      &-             & 0.691008        			\\ \hline
19   & 0.691433        &-           & 0.691425        			\\ \hline
20  & 0.691775        &-           & 0.691783        			\\ \hline
\end{tabular}
\captionof{table}{Numerical results: Heuristic}
\label{table:heuristic}
\end{table}

\begin{itemize}
    \item We observe that the randomized primal-dual analysis of our paper and the analysis of Mahdian and Yan~\cite{stoc/MahdianY11} share numerically close competitive ratios for small $m,n$. We conjecture the two analysis would lead to the same competitive ratio when $m,n \to \infty$, supported by further numerical experiments. Apart from the numerical results reported in Section~\ref{sec:num}, we also adapt the local search algorithm of Section~\ref{subsec:num_upper} as a heuristic to solve (LP: Lower Bound)\footnote{Here, we maintain a set $S$ of grid paths $\vect{b}$ for implementing the first and second family of constraints. In each iteration, we examine all local perturbations of the current set and update $S$.}. Although the heuristic algorithm only provides an upper bound of (LP: Lower Bound) as we restrict the first family of constraints to a proper subset, we believe the numerical numbers are meaningful and are close to the optimal value of the program. Indeed, for $n \le 10$, our heuristic algorithm is able to find the optimal value of (LP: Lower Bound). The numerical results are summarized in Table~\ref{table:heuristic}.
\item There remains a relatively large gap between the barrier of the current analysis and the upper bound of Ranking's competitive ratio in the random arrival model. Both analysis of our work and Mahdian and Yan~\cite{stoc/MahdianY11} is upper bounded by $0.703$, while the best known upper bound of Ranking's competitive ratio is $0.724$ by Chan et al.~\cite{sicomp/ChanCWZ18}.
We suggest to study the independent random arrival model when $m=2$ as the first step towards a tight competitive ratio of Ranking. A concrete question is to first solve the optimization (LP: Lower Bound) for $m=2, n \to \infty$ in a closed form.
\item What is the minimal assumption on the arrival order that allows a better than $1-1/e$ competitive ratio for the online bipartite matching problem? Our independent random arrival model nicely bridges the adversarial arrival model and the classical random arrival model, but it unavoidably requires exponential number of different orders. It is interesting to explore whether polynomial number of different arrival orders suffice.
\end{itemize}